\documentclass[10pt]{article}
\usepackage{caption}

\usepackage{subcaption}
\usepackage[utf8]{inputenc}

\usepackage{amssymb}
\usepackage{amsfonts}
\usepackage{amsmath}
\usepackage{amsthm}
\usepackage{mathtools}

\usepackage{graphicx}
\usepackage{indentfirst}
\usepackage{cite}
\usepackage{enumerate}
\usepackage{url}

\usepackage{lmodern}
\usepackage{newpxtext}

\usepackage[symbol]{footmisc}

\usepackage{tikz}
\definecolor{darkblue}{rgb}{0.15,0.35,0.55}
\definecolor{reddish}{rgb}{.8, 0.2, 0.2}

\usepackage{pgfplots}
\definecolor{plotblue}{RGB}{0,120,200}
\definecolor{plotgreen}{RGB}{0,155,130}
\definecolor{plotorange}{RGB}{240,120,50}
\definecolor{plotmagenta}{RGB}{240,50,120}
\definecolor{plotgray}{RGB}{128,128,128}
\definecolor{plotcyan}{RGB}{50,190,240}
\definecolor{plotred}{RGB}{205,50,15}
\usepgfplotslibrary{fillbetween}

\usepackage[pdfpagelabels, linktocpage=true]{hyperref}
\hypersetup{
colorlinks=true,
citecolor=darkblue,
linkcolor=reddish,
urlcolor=darkblue,
pdfauthor={},
pdftitle={},
pdfsubject={}
}

\long\def\ca#1\cb{} 

\newcommand{\becs}{\begin{cases}}
	\newcommand{\bem}{\begin{matrix}}

		\newcommand{\bra}[1]{\langle#1|}

		\newcommand{\dya}[1]{|#1\rangle\langle#1|}
		\newcommand{\dyad}[2]{|#1\rangle\langle#2|}
		
		\newcommand{\encs}{\end{cases}}
	\newcommand{\enm}{\end{matrix}}

\newcommand{\inp}[1]{\langle#1|#1\rangle }
\newcommand{\inpV}[2]{\langle#1,#2\rangle }
\newcommand{\inpd}[2]{\langle#1|#2\rangle }
\newcommand{\ket}[1]{|#1\rangle }

\newcommand{\ot}{\otimes }

\newcommand{\Tr}{{\rm Tr}}

\newcommand{\AC}{{\mathcal A}}
\newcommand{\BC}{{\mathcal B}}

\newcommand{\DC}{{\mathcal D}}
\newcommand{\EC}{{\mathcal E}}

\newcommand{\HC}{{\mathcal H}}
\newcommand{\IC}{{\mathcal I}}
\newcommand{\JC}{{\mathcal J}}
\newcommand{\KC}{{\mathcal K}}
\newcommand{\LC}{{\mathcal L}}
\newcommand{\MC}{{\mathcal M}}
\newcommand{\NC}{{\mathcal N}}
\newcommand{\OC}{{\mathcal O}}
\newcommand{\PC}{{\mathcal P}}
\newcommand{\QC}{{\mathcal Q}}

\newcommand{\SC}{{\mathcal S}}

\newcommand{\eB}{\textbf{e}}

\newcommand{\xB}{\textbf{x}}

\newcommand{\Rbb}{\mathbb{R}}

\newcommand{\al}{\alpha }
\newcommand{\bt}{\beta }
\newcommand{\gm}{\gamma }

\newcommand{\dl}{\delta }

\newcommand{\lm}{\lambda }
\newcommand{\Lm}{\Lambda }

\newcommand{\sg}{\sigma }

\theoremstyle{remark}
\newtheorem{theorem}{Theorem}
\newtheorem{lemma}{Lemma}

\usepackage[linesnumbered]{algorithm2e}

\usetikzlibrary{arrows}
\usetikzlibrary{decorations.pathreplacing,calligraphy}
\usepackage{verbatim}

\usepackage{bbm}
 
\usepackage[affil-it]{authblk}
\usepackage[normalem]{ulem}

\author{Vikesh Siddhu and John Smolin}
\title{Optimal one-shot entanglement sharing}
\affil{IBM Quantum, IBM T.J. Watson Research Center, Yorktown Heights, NY 10598, USA}
\date{06 Oct 2023}
\begin{document}

\maketitle
\begin{abstract}
    Sharing entanglement across quantum interconnects is fundamental for
    quantum information processing.  We discuss a practical setting where this
    interconnect, modeled by a quantum channel, is used once with the aim
    of sharing high fidelity entanglement.
    For any channel, we provide methods to easily find both this maximum
    fidelity and optimal inputs that achieve it. Unlike most metrics for
    sharing entanglement, this maximum fidelity can be shown to be
    multiplicative.  This ensures a complete understanding in the sense that
    the maximum fidelity and optimal inputs found in our one-shot setting
    extend even when the channel is used multiple times, possibly with other
    channels.
    Optimal inputs need not be fully entangled. We find the minimum
    entanglement in these optimal inputs can even vary discontinuously with
    channel noise.
    Generally, noise parameters are hard to identify and remain unknown for
    most channels. However, for all qubit channels with qubit environments, we
    provide a rigorous noise parametrization which we explain in-terms of
    no-cloning.
    This noise parametrization and a channel representation we call the
    standard Kraus decomposition have pleasing properties that make them both
    useful more generally.
\end{abstract}
\tableofcontents
\newpage

\section{Introduction}

Quantum computation and communication requires faithful transmission of quantum
information between various separated parties.
These parties may be closely separated quantum computing nodes or widely
separated receivers and transmitters of quantum states. The former appear in
models of a quantum intranet~\cite{chow21} while the latter appear in
discussions of a quantum internet~\cite{Kimble08, WehnerElkoussEA18}.
Noise in these, and other such setups hinders their use. A dominant source of
noise is the quantum interconnect carrying quantum information between parties.
This interconnect is modeled mathematically by a quantum channel, a completely
positive trace preserving map. Quantum information sent and processed across
this channel is equivalent to entanglement shared and processed using the
channel~\cite{BennettBrassardEA93}.
Without investigating methods, metrics, protocols and characteristics of
sharing entanglement across quantum channels, our understanding and ability to
control and scale quantum computation and communication remains partial.

The most well studied setting for sharing entanglement
allows asymptotically many channel uses~\cite{BennettDiVincenzoEA96, Wilde17}.
Across all channels used together, local pre- and post-processing of
entanglement is allowed along with classical communication from channel input
to output.
Using these allowed operations, the largest number of fully entangled states,
per channel use, shared with asymptotically vanishing error defines the quantum
capacity of the channel. Studies of this metric reveal that while
theoretically beautiful~\cite{Lloyd97, Shor02a, CaiWinterEA04, Devetak05,
DevetakShor05} a channel's quantum capacity is hard to compute and non-trivial
to understand in general~\cite{SmithYard08, CubittElkoussEA15}. Both these
features come from super-additivity. Super-additivity of quantum capacity
implies that the quantum capacity of several channels used jointly is not
completely specified by the quantum capacity of each channel~\cite{Siddhu21,
LeditzkyLeungEA22}.

Asymptotic channel capacities provide rich conceptual and practical
difficulties. For these reasons, it is desirable to study entanglement
transmission with as little encoding and decoding as possible.
The simplest setting here is a single use of a channel~(which can itself be
joint uses of many channels) with no post-processing. This setting need
not allow sharing of noiseless entanglement. Thus, one may define a metric for
sharing entanglement with some acceptable level of noise. One such metric,
called the one-shot quantum capacity, is roughly the largest fully entangled
state than can be shared across a channel with at most a fixed, but arbitrary
error~\cite{BuscemiDatta10}.
This one-shot capacity, its connection to asymptotic capacities, and method for
understanding and achieving these have been recently
explored~\cite{DattaHsieh11, DattaHsieh13, MatthewsWehner14, BeigiDattaEA16,
TomamichelBertaEA16, PfisterRolEA18, AnshuJainEA19, WangFangEA18,
SalekAnshuEA20, NakataWakakuwaEA21, KhatriWilde20}.
However, we don't fully understand notions of additivity for this capacity;
ways of computing and explicit protocols for achieving the one-shot capacity
are not completely known.

A key metric in the one-shot setting is the highest fidelity between the state
shared across the channel and a maximally entangled
state~\cite{BennettBrassardEA96, BennettDiVincenzoEA96}. This fidelity
characterizes optimal performance of various teleportation based
tasks~\cite{HorodeckiHorodeckiEA99}.  The optimal fidelity between a pure
entangled state shared across the channel and a maximally entangled state is
known~\cite{VerstraeteVerschelde03, PalBandyopadhyay18}.  Surprisingly, the
optimal pure state input need not be maximally entangled which is consistent
with fidelity not being an entanglement monotone.

The one-shot setting is augmented by post-processing using one round of local
operations and two-way classical
communication~(2-LOCC)~\cite{VerstraeteVerschelde03a, BandyopadhyayGhosh12,
PalBandyopadhyayEA14, PalBandyopadhyay18}. However, in this setting it is
unknown if the optimal fidelity is multiplicative~(analog of additivity in this
setting). There is no known method for computing or explicit protocol for
achieving this optimal fidelity in general.  The only exception is qubit
channels, where optimal protocols use pure state inputs and don't require
2-LOCC~\cite{VerstraeteVerschelde03a, PalBandyopadhyayEA14}. Surprisingly, the
behaviour of such optimal protocols for the simplest of qubit channels is not
fully known. 

One way to understand a metric for sharing entanglement across a specific
channel is to study variation in the metric with the amount of noise in the
channel. Surprisingly, even for the simplest qubit channels, noise parameters
are only partially understood.

\begin{figure}[ht]
    \centering
    \includegraphics[]{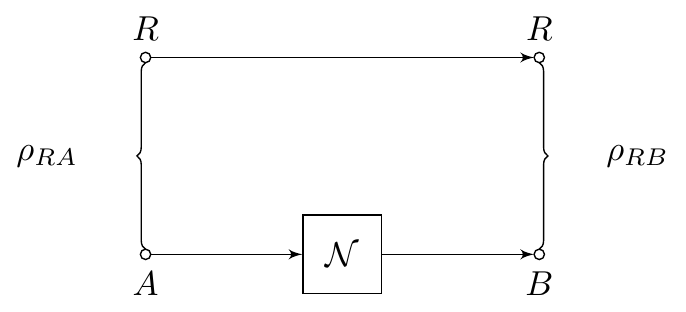}
    \caption{Diagram representing one-shot entanglement passing.}
    \label{fig:setup}
\end{figure}
\newpage

{\bf Results.} In this work, we introduce and solve the problem for sharing
entanglement in a one-shot setting where an arbitrary mixed state $\rho_{RA}$
may be prepared across a reference system $R$ and channel input $A$. This input
is sent via a fixed channel $\NC:A \mapsto B$~(see Fig.~\ref{fig:setup}) to
achieve the maximum fidelity $\OC(\NC)$ between the channel output $\rho_{RB}$
and a maximally entangled state across $R$ and $B$.
We reformulate $\OC(\NC)$ via a semi-definite
program~\cite{VandenbergheBoyd96}.  Our first main result is to express
$\OC(\NC)$ in two useful ways ~(see Sec.~\ref{sec:BasicQ} with
Th.~\ref{th:thOne})
First, using what we define~(in Sec.~\ref{sec:standard}) as a channel's
standard Kraus decomposition, and second, in terms of the operator norm of a
channel's Choi–Jamio\l{}kowski operator. Next, we show the maximum
fidelity $\OC$ is multiplicative~(see Th.~\ref{th:multi}).
Not only can $\OC(\NC)$ be achieved using pure states but, in certain
cases, it can also be achieved using a variety of mixed states. We give a
recipe to construct these pure and mixed states.
For all extremal~(see definition below eq.~\eqref{eq:qChanCVX}) qubit channels,
we compute optimal inputs and the minimum amount of entanglement $\EC$ required
to create these inputs. 
Somewhat surprisingly, the minimum entanglement
$\EC$ is found to be discontinuous in these noise parameters. Typically, $\EC$
is less than its maximal value of one, but $\OC$ is high enough for the channel
to be useful for teleportation, even if the channel has no quantum
capacity~(see Sec.~\ref{sec:glnQbitNoise}).
For very noisy qubit Pauli channels we find separable inputs that achieve the
same fidelity as maximally entangled ones found previously~(see
Sec.~\ref{sec:pauliCase}).
We also find optimal inputs for qutrit channels have a much richer structure
than qubit channels~(see Sec.~\ref{sec:qtrit}).
Noise parameters for general channels remain unknown. We find rigorous noise
parameters for all extremal qubit channels~(see Sec.~\ref{sec:extremeQubit}), a
result which may be of independent interest.
Unlike other metrics in settings for entanglement sharing, $\OC$ is
multiplicative. Thus, even when a channel $\NC$ is used multiple times,
possibly with other channels, its maximum fidelity $\OC(\NC)$ fully
characterizes its ability for sharing high fidelity entanglement
without post-processing. 
Our results also give rigorous lower bounds on entanglement fidelities that can
be achieved when allowing for multiple rounds of 2-LOCC. These bounds are tight
for one round of 2-LOCC using qubit channels. Characterization of the noise
parameters for all extremal qubit channels presented here pave the way for a
stronger understanding of quantum channels and quantum protocols across
channels.

\section{Preliminaries}

Let $\xB$ denote a vector in $n$-dimensional real space, $\Rbb^n$, $\xB_i$
denote the $(i+1)^{\text{th}}$ coordinate of $\xB$, and $|\xB_i|$ denote its
absolute value. Coordinates of $\xB$ rearranged in decreasing order give
$\xB^{\downarrow}$, a vector satisfying $\xB^{\downarrow}_0 \geq
\xB^{\downarrow}_1 \geq \dots \geq \xB^{\downarrow}_{n-1}$. Euclidean norm of
$\xB$, $|\xB| := \sqrt{\sum_i \xB^2_i}$.
Let $\ket{\psi}$ denote a ket in a Hilbert space $\HC$ of finite dimension $d$
and $| \ket{\psi} | := \sqrt{\inp{\psi}}$ denote its norm. A pure quantum
state is represented by a ket with unit norm.
Let $\LC(\HC)$ denote the space of linear operators on $\HC$. For any two
quantum states $\ket{\psi}$ and $\ket{\phi}$, the dyad $\dyad{\psi}{\phi} \in
\LC(\HC)$ and the projector onto $\ket{\psi}, \dya{\psi} \in \LC(\HC)$. The
Frobenius inner product between two operators $N$ and $O$ in $\LC(\HC)$,
\begin{equation}
    \inpV{N}{O} := \Tr (N^{\dag} O),
    \label{eq:frob}
\end{equation}
where $N^{\dag}$ represents the adjoint~(conjugate transpose) of $N$.  
A Hermitian operator $H \in \LC(\HC)$, satisfying $H=H^{\dag}$, represents
an observable. This operator has an eigendecomposition,
\begin{equation}
    H = \sum_i \xB_i \dya{\psi_i},
    \label{eq:spectral}
\end{equation}
where $\xB_i \in \Rbb$ is an eigenvalue of $H$ corresponding to eigenvector
$\ket{\psi_i}$ and the collection of eigenvectors $\{\ket{\psi_i}\}$ form an
orthonormal basis of $\HC$, $\inpd{\psi_i}{\psi_j} = \dl_{ij}$, where
$\dl_{ij}$ is the Kronecker delta function. Support of $H$ is the subspace
spanned by its eigenvectors with non-zero eigenvalues. 
In~\eqref{eq:spectral}, if $\xB_i \geq 0$ for all $i$, then we say $H$ is
positive semi-definite~(PSD), $H \succeq 0$.
An optimization, over PSD matrices, of the form
\begin{align}
    \begin{aligned}
        \text{maximize} \; & \Tr(A_0H) & \\
        \text{subject to} \; & \Tr(A_iH) = c_i, & \forall \; 1 \leq i \leq n, \\
        \text{and} \; & H \succeq 0,
    \end{aligned}
\end{align}
where $A_i$ are Hermitian, is called a semi-definite program~(see Sec.1.2.3
in~\cite{Watrous18} and citations to and within~\cite{VandenbergheBoyd96}).
Square root of a PSD operator $H$, $\sqrt{H}$, is obtained by replacing
$\xB_i$ in~\eqref{eq:spectral} with $\sqrt{\xB_i}$.  For any operator $O \in
\LC(\HC)$,
\begin{equation}
    || O || := \underset{| \ket{\psi} | \leq 1}{\max} | O \ket{\psi} |,
    \quad
    || O ||_1 := \Tr(\sqrt{OO^{\dag}}), \quad \text{and} \quad
    || O ||_2 := \sqrt{\Tr(OO^{\dag})}, 
    \label{eq:Norm}
\end{equation}
denote the spectral norm, the trace norm, and the Frobenius norm, respectively.
For $H$ in~\eqref{eq:spectral}, 
\begin{equation}
    ||H|| = |\xB^{\downarrow}_0|, 
    \quad
    ||H||_1 = \sum_i |\xB_i|, \quad \text{and} \quad 
    ||H||_2 = |\xB|.
    \label{eq:HNorm}
\end{equation}

A density operator $\rho \in \LC(\HC)$ is a positive semi-definite operator
with unit trace, $\Tr(\rho) = 1$, it represents a mixed quantum state. 
Its von-Neumann entropy,
\begin{equation}
    S(\rho) = - \Tr (\rho \log \rho),
\end{equation}
where $\log$ is base $2$.
The fidelity between two density operators $\rho$ and $\sg$,
\begin{equation}
    F(\rho,\sg) := || \sqrt{\rho} \sqrt{\sg} ||_1.
    \label{eq:Fidelity}
\end{equation}

Let $\HC_A$ and $\HC_B$ be two Hilbert spaces of dimensions $d_A$ and $d_B$,
respectively, and $\HC_{AB}$ denote the tensor product space $\HC_A \ot \HC_B$.
Given a pure state $\ket{\psi}_{AB} \in \HC_{AB}$, density operators
\begin{equation}
    \psi_A = \Tr_B (\dya{\psi}) \quad \text{and} \quad
    \psi_B = \Tr_A (\dya{\psi})
\end{equation}
denote the partial trace of $\dya{\psi}$ over $\HC_B$ and $\HC_A$,
respectively.  The entanglement of formation of a pure state $\ket{\psi}_{AB}$,
\begin{equation}
    E_f(\ket{\psi}_{AB}) = S(\psi_A),
    \label{eq:entEntro0}
\end{equation}
and for a mixed state $\rho_{AB}$,
\begin{equation}
    E_f(\rho_{AB}) = \text{min} \; \sum_i p_i E_f(\ket{\psi_i}_{AB}),
    \label{eq:entEntro}
\end{equation}
is the minimum average entanglement $E_f$ over all pure state decompositions,
$\rho_{AB} = \sum_i p_i \dya{\psi_i}$, $p_i \geq 0$ and $\sum_i p_i = 1$.

Let $\AC = \{ \ket{a_i}\}$ and $\BC = \{ \ket{b_j}\}$ be orthonormal bases,
of $\HC_A$ and $\HC_B$, respectively, i.e., 
\begin{equation}
    \inpd{a_i}{a_{j}} = \inpd{b_i}{b_j} = \dl_{ij}.
    \label{eq:abBasis}
\end{equation}
Using these bases $\AC$ and
$\BC$ we can represent any linear operator $L : \HC_A \mapsto \HC_B$ as a
matrix,
\begin{equation}
    L = \sum_{ij} [L]_{ij} \dyad{b_i}{a_j},
\end{equation}
with elements $[L]_{ij}$. We can define two basis dependent linear maps, 
\begin{equation}
    L^* = \sum_{ij} [L]_{ij}^* \dyad{b_i}{a_j}, \quad \text{and} \quad
    L^T = \sum_{ij} [L]_{ij} \dyad{a_j}{b_i},
    \label{eq:trConj}
\end{equation}
representing complex conjugate and transpose, respectively. In contrast to
$L^*$ and $L^T$, the adjoint $L^{\dag} = (L^*)^T = (L^T)^*$ is basis
independent.
If $\HC_A$ and $\HC_B$ have the same dimension $d$, then one can choose $\AC$
and $\BC$ to be the same, say the standard basis $\{\ket{i}\}$, and
construct an identity map $I_{A \gets B}: \HC_B \mapsto \HC_A$,
\begin{equation}
    I_{A \gets B} \ket{i}_{B} = \ket{i}_A.
    \label{eq:idOp}
\end{equation}
This subscript notation $A \gets B$ is dropped shortly after defining how the
identity map above is used to map a ket $\ket{\phi}_B \in \HC_B$, an operator
$O_B \in \LC(\HC_B)$ and part of an operator $L_{AB} \in \LC(\HC_{AB})$ to 
\begin{equation}
    \ket{\phi}_A = I_{A \gets B}\ket{\psi}_B, \quad
    O_A = I_{A \gets B} O_{B} I_{B \gets A}, \quad \text{and} \quad
    L_{AA} = (I_{A \gets B} \ot I_A) L_{BA} (I_{B \gets A} \ot I_A),
\end{equation}
respectively, here $I_A$ is identity on the $\HC_A$ space. Later, these
mappings are done implicitly by simply replacing the subscripts in an obvious
way.

\subsection{Operator-Ket duality}
\label{sec:opKetD}

{\em Operator-ket} duality is the concept of fixing an orthonormal basis $\AC =
\{\ket{a_i}\}$ of $\HC_A$ and using an un-normalized maximally entangled
state on $\HC_A \ot \HC_A$,
\begin{equation}
    \ket{\gm}_{AA} = \sum_{i} \ket{a_i} \ot \ket{a_i},
    \label{eq:maxEntState}
\end{equation}
to associate with any linear operator $K : \HC_A
\mapsto \HC_B$ a ket, $\ket{\psi}_{AB} = (I_A \ot K) \ket{\gm}$,
obtained by acting $K$ on one-half of $\ket{\gm}$. Conversely, for fixed
orthonormal basis $\AC$, one associates with any ket $\ket{\psi}_{AB}$, a linear
operator
\begin{equation}
    K = \sum_i \dyad{\chi_i}{a_i}, \quad \text{where} \quad
    \ket{\chi_i}_B = (\bra{a_i}_A \ot I_B) \ket{\psi}_{AB}.
    \label{eq:mapKet2}
\end{equation}
In analogy to the discussion above, fixing an orthonormal basis $\BC
= \{\ket{b_j}\}$ of $\HC_B$ one associates with the ket $\ket{\psi}_{AB}$ an
operator $L : \HC_B \mapsto \HC_A$.  This operator $L = K^T$ where the
transpose operation is taken using basis $\AC$ and $\BC$ as described
in~\eqref{eq:trConj}.

In what follows, we use the notation $\ket{K} \in \HC_{AB}$ for a ket
associated with the operator $K : \HC_A \mapsto \HC_B$ through the operator-ket
duality above where basis $\AC$ is fixed. This ket and operator pair satisfy
\begin{equation}
    \ket{K}_{AB} = (I \ot K) \ket{\gm}_{AA}.
    \label{eq:mapKet}
\end{equation}
For any two maps $K$ and $K'$ from $\HC_A$ to $\HC_B$ and associated kets
$\ket{K}_{AB}$ and $\ket{K'}_{AB}$, respectively, one can show that
\begin{equation}
    \inpV{K}{K'} = \inpd{K}{K'}.
    \label{eq:inpDual}
\end{equation}
Using the orthonormal basis $\BC$ of $\HC_B$, one can associate with $K^{\dag}:
\HC_B \mapsto \HC_A$ the ket $\ket{K^{\dag}}_{BA}$. In this ket, swapping the
spaces $\HC_A$ and $\HC_B$~(see discussion below~\eqref{eq:idOp}) gives
$\ket{K^{\dag}}_{AB}$ which then satisfies
\begin{equation}
    \ket{K^{\dag}}_{AB} = \ket{K}_{AB}^*
    \label{eq:inpDual2}
\end{equation}
where complex conjugation of any ket $\ket{\chi}_{AB} = \sum_{ij}c_{ij}
\ket{a_i} \ot \ket{b_j}$, is defined using basis $\AC$ and $\BC$ as
$\ket{\chi}_{AB}^* = \sum_{ij} c^*_{ij} \ket{a_i} \ot \ket{b_j}$.

\section{Quantum channels}
\label{sec:qChan}

Let $\HC_A, \HC_B,$ and $\HC_C$ be three Hilbert spaces and $V: \HC_A \mapsto
\HC_B \ot \HC_C$ be an isometry, i.e., $V^{\dag} V = I_A$. This isometry
defines a pair of quantum channels $\NC$ and $\NC^c$, i.e., a pair of
completely positive trace preserving~(CPTP) maps with superoperators
\begin{equation}
    \NC(O) = \Tr_C (V O V^{\dag}) \quad \text{and} \quad
    \NC^c(O) = \Tr_B (V O V^{\dag}),
\end{equation}
taking $O \in \LC(\HC_A)$ to $\LC(\HC_B)$ and $\LC(\HC_C)$, respectively.  The
quantum channel $\NC$ is called degradable and $\NC^c$ anti-degradable if there
exists a quantum channel $\DC$ such that $\DC \circ \NC =
\NC^c$~\cite{DevetakShor05}.

Let $\IC_A$ be the identity map from $\LC(\HC_A)$ to itself.  Using an
un-normalized maximally entangled state $\ket{\gm}_{AA}$~\eqref{eq:maxEntState}
we define the Choi–Jamio\l{}kowski~\cite{Jamiokowski72, Choi75} operator of the
linear map $\NC$ as
\begin{equation}
    J^{\NC}_{AB} = \IC_A \ot \NC (\dya{\gm}) = 
    \sum_{ij} \dyad{a_i}{a_j} \ot \NC(\dyad{a_i}{a_j}).
    \label{eq:cj}
\end{equation}
This operator contains all information about $\NC$. For instance,
\begin{equation}
    \NC(\dyad{a_i}{a_j}) = 
    (\bra{a_i} \ot I_B ) J^{\NC}_{AB} (\ket{a_j} \ot I_B),
    \label{eq:cj2}
\end{equation}
$\NC$ is completely positive~(CP) if and only if $J^{\NC}_{AB}$ is positive
semi-definite, and 
\begin{equation}
    \Tr_B(J^{\NC}_{AB}) = I_A
\end{equation}
if and only if $\NC$ is trace-preserving, $\Tr \big( \NC(O) \big) =
\Tr(O)$ for all $O$. 
Equivalently, a linear map $\NC: \LC(\HC_A) \mapsto \LC(\HC_B)$ is CP
if and only if it can be written in the form
\begin{equation}
    \NC(O) = \sum_i K_i O K_i^{\dag},
    \label{eq:kraus}
\end{equation}
where $K_i : \HC_A \mapsto \HC_B$ is a linear operator, and the collection
$\{K_i\}$ are called Kraus operators. The map in~\eqref{eq:kraus} is trace
preserving when these Kraus operators satisfy $\sum_i K_i^{\dag} K_i = I_A$.
When $\NC$ is unital, i.e., $\NC(I_A) = I_B$, the Kraus operators satisfy
$\sum_i K_i K_i^{\dag} = I_B$.  If $\HC_A$ and $\HC_B$ have the same dimension,
then they are isomorphic to one another and can be denoted by $\HC$. If these
Kraus operators on $\HC$ are Hermitian operators~(or normal operators) then the
channel is automatically unital.

\subsection{A standard Kraus decomposition}
\label{sec:standard}

For a given channel $\NC: \LC(\HC_A) \mapsto \LC(\HC_B)$, the set of Kraus
operators is not unique. However, one can construct what can be called a {\em
standard Kraus decomposition} with some pleasing properties. 

Consider the eigendecomposition of the Choi–Jamio\l{}kowski operator
in~\eqref{eq:cj},
\begin{equation}
    J^{\NC}_{AB} = \sum_i \eB^{\downarrow}_i \dya{L_i},
    \label{eq:choi2}
\end{equation}
where eigenvalues $\eB^{\downarrow}_0 \geq \eB^{\downarrow}_1 \geq \dots \geq
\eB^{\downarrow}_{d_Ad_B-1}\geq 0$ and eigenvectors $\{ \ket{L_{i}} \}$ form an
orthonormal basis of $\HC_{AB}$,
\begin{equation}
    \inpd{L_{i}}{L_{j}} = \dl_{ij}.
    \label{eq:good}
\end{equation}
Applying operator-ket duality using orthonormal basis $\AC = \{\ket{a_i}\}$ to
kets $\{ \ket{L_{i}} \}$ results in a collection of orthonormal operators $\{
    L_{i} \}$ that map $\HC_{A}$ to $\HC_B$~(see Sec.~\ref{sec:opKetD}). Using
these operators define $K_{i}: \HC_A \mapsto \HC_B$,
\begin{equation}
    K_{i}:= \sqrt{\eB^{\downarrow}_i} L_{i},
    \label{eq:KijDef}
\end{equation}
and notice from operator-ket duality we get
\begin{equation}
    \ket{K_{i}}:= \sqrt{\eB^{\downarrow}_i} \ket{L_{i}}.
    \label{eq:KijDef2}
\end{equation}
\begin{lemma}
    \label{lm:stdKraus}
Operators $\{ K_{i} \}$ form a Kraus decomposition of $\NC$,
\begin{equation}
    \NC(O) = \sum_{i} K_{i} O K_{i}^{\dag}.
\end{equation}
\end{lemma}
\begin{proof}
    In~\eqref{eq:choi2} use~\eqref{eq:mapKet} to obtain 
    \begin{align}
        J^{\NC}_{AB} &= \sum_{i} \eB^{\downarrow}_i (I_A \ot L_{i})
        \dya{\gm} (I_A \ot L_{i})^{\dag}\\
        &= \sum_{i} (I_A \ot K_{i}) \dya{\gm} (I_A \ot K_{i})^{\dag},
    \end{align}
    where the second inequality uses~\eqref{eq:KijDef}. This second inequality,
    together with~\eqref{eq:cj2} gives,
    \begin{equation}
        \NC(\dyad{a_k}{a_l}) = \sum_{i} K_{i} (\dyad{a_k}{a_l}) K_{i}^{\dag} 
    \end{equation}
    This equality, together with linearity of $\NC$ proves this lemma.
\end{proof}
Using~\eqref{eq:choi2}, \eqref{eq:good}, \eqref{eq:KijDef},
and~\eqref{eq:KijDef2} one can show that the Kraus operators $\{ K_{i}\}$
satisfy
\begin{equation}
    \inpV{K_{i}}{K_{j}} = \inpV{K_{i}}{K_{i}} \dl_{ij}
    \quad \text{and} \quad
    \inpV{K_{i}}{K_{i}} \geq \inpV{K_{j}}{K_{j}},
    \label{eq:canKraus}
\end{equation}
where $i \leq j$ and we use $\inpV{K_{i}}{K_{i}} = \eB^{\downarrow}_i$.
In addition to being orthogonal and ordered in the way captured by the above
equation, the Kraus operators $\{K_i\}$ have several other useful properties.
The total number of non-zero operators $\{K_i\}$ is the rank of the
Choi-Jamio\l{}kowsi operator $\JC_{AB}^{\NC}$. This rank is the minimum number
of Kraus operators required to represent the channel $\NC$. When the
eigenvalues of $J_{AB}^{\NC}$ are distinct, the norm $\inpV{K_i}{K_i}$ of each
Kraus operator is simply the $(i+1)^{\text{th}}$ largest eigenvalue of
$\JC_{AB}^{\NC}$.  From these Kraus operators, one can obtain the
Choi-Jamio\l{}kowsi operator~\eqref{eq:cj},
\begin{equation}
    \JC_{AB}^{\NC} = \sum_i \dya{K_i},
    \label{eq:choiCan}
\end{equation}
where we have applied operator-ket duality~(see Sec.~\ref{sec:opKetD}) to
convert operators $K_i : \HC_A \mapsto \HC_B$ to kets $\ket{K_i} \in \HC_{AB}$
using basis $\AC = \{\ket{a_i}\}$. Notice $\ket{K_i}$ is an un-normalized
eigenvectors of $\JC_{AB}^{\NC}$ with eigenvalue $\inpV{K_i}{K_i}$.  We call
$\{K_i\}$ in Lemma~\eqref{lm:stdKraus} to be a standard Kraus decomposition. 

\subsection{Dual channel}
\label{sec:dual}

Given a map $\NC: \LC(\HC_A) \mapsto \LC(\HC_B)$, its dual $\NC^{\dag}:
\LC(\HC_B) \mapsto \LC(\HC_A)$ is defined via~\footnote{This definition of dual
map~\eqref{eq:dualInp}, common in quantum information~(see Def.(6.2)
in~\cite{Holevo12} or below eq.(1.44) in~\cite{Wolf12}), differs from another,
$\inpV{\NC^{\dag}(O)}{\rho} = \inpV{O}{\NC(\rho)}$, found in mathematics
literature.  The two definitions coincide for maps satisfying,
$\NC(\rho^{\dag}) = \big( \NC(\rho) \big)^{\dag}$, but can differ when this
property is not satisfied. For example if $\NC(\rho) = c\rho$, and $c$ complex
then the two definitions give different dual maps.} 
\begin{equation}
    \Tr\big( \NC^{\dag}(O)\rho \big) = \Tr \big( O\NC(\rho) \big),
    \label{eq:dualInp}
\end{equation}
where $\rho \in \LC(\HC_A)$ and $O \in \LC(\HC_B)$.
A quantum channel $\NC$ evolves a quantum state $\rho$ and its dual channel
$\NC^{\dag}$ evolves an observable $O$.  The right side of the above equality
represents the expectation value of the evolved quantum state $\NC(\rho)$ with
respect to a fixed observable $O$ while the left side of the equality gives the
expectation value of a fixed state $\rho$ with respect to the evolved
observable $\NC^{\dag}(O)$.
If $\NC$ is CP and has Kraus decomposition~\eqref{eq:kraus} then $\NC^{\dag}$
is also CP with Kraus operators $\{K_i^{\dag}\}$, and if $\NC$ is
trace-preserving then $\NC^{\dag}$ is unital~(see Ch.6
in~\cite{Holevo12}).
A CP map $\NC$ with standard Kraus operators $\{K_i\}$ has dual map
$\NC^{\dag}$ with standard Kraus operators $\{K_i^{\dag}\}$ since
\begin{equation}
    \inpV{K_i^{\dag}}{K_j^{\dag}} = (\inpV{K_i}{K_j})^*.
    \label{eq:choiDcan}
\end{equation}
The Choi-Jamio\l{}kowsi operator~\eqref{eq:cj} of the dual channel,
\begin{equation}
    J^{\NC^{\dag}}_{BA} = \sum_i \dya{K_i^{\dag}},
    \label{eq:choiD}
\end{equation}
where $\{ \ket{K_i^{\dag}}\}$ in $\HC_{BA}$ are defined via operator-ket
duality using basis $\BC = \{ \ket{b_j} \}$. 
We aim to compare $J^{\NC^{\dag}}_{BA}$, an operator on $\HC_{BA}$ with
$J^{\NC}_{AB}$, an operator on a different space $\HC_{AB}$. To carry out the
comparison, interchange
$B$ and $A$ in~\eqref{eq:choiD} and use~\eqref{eq:inpDual2}, \eqref{eq:choiCan}
to get
\begin{equation}
    J^{\NC}_{AB} = (J^{\NC^{\dag}}_{AB})^*.
    \label{eq:choiDchoi}
\end{equation}
The Choi-Jamio\l{}kowsi operator of a channel and its dual can be taken to be
complex conjugates of one another.

\subsection{Extremal qubit channels}
\label{sec:extremeQubit}

The set of quantum channels from $\LC(\HC_A)$ to $\LC(\HC_B)$ is convex, i.e.,
if $\NC$ and $\MC$ are quantum channels then 
\begin{equation}
    \KC = \lm \NC + (1-\lm) \MC,
    \label{eq:qChanCVX}
\end{equation}
is a quantum channel for any $0 \leq \lm \leq 1$. Any quantum channel $\KC$ is
extremal, i.e., it is an extreme point of the set of quantum channels, if
equality of the type~\eqref{eq:qChanCVX} holds only when $\lm=0$ or $\lm=1$, or
the only channels $\NC$ and $\MC$ satisfying the equality both equal $\KC$.

A quantum channel $\NC : \LC(\HC_A) \mapsto \LC(\HC_B)$ is called a qubit
channel when $\HC_A$ and $\HC_B$ are two-dimensional. 
For these two dimensional spaces, we can use the standard basis $\{\ket{i}\}$,
where $i \in \{0,1\}$, to define Pauli operators,
\begin{equation}
    X = \dyad{0}{1} + \dyad{1}{0}, \quad 
    Y = -i\dyad{0}{1} + i \dyad{1}{0}, \quad \text{and} \quad  
    Z = \dya{0} - \dya{1}.
    \label{eq:pauliOp}
\end{equation}

Extreme points of qubit channels are studied in various
works~\cite{FujiwaraAlgoet98, NiuGriffiths99, RuskaiSzarekEA02,
WolfPerezGarcia07, FriedlandLoewy14}. A qubit channel is extremal if it has a
single Kraus operator, given by a unitary operator, or it has two Kraus
operators, each not proportional to a unitary operator~(see Cor.~15
in~\cite{FriedlandLoewy14}). Up to local unitaries at the channel input and
output, a qubit channel $\NC$ with two Kraus operators can be written
as~\cite{RuskaiSzarekEA02}
\begin{equation}
    \NC(O) = K_0 O K_0^{\dag} + K_1 O K_1^{\dag},
    \label{eq:twoKraus}
\end{equation}
where, 
\begin{equation}
    K_0 = \begin{pmatrix}
        \cos (\frac{v-u}{2}) & 0 \\
        0 & \cos (\frac{v+u}{2})  
    \end{pmatrix}, \quad
    K_1 = \begin{pmatrix}
        0 & \sin (\frac{v+u}{2}) \\
        \sin (\frac{v-u}{2}) & 0
    \end{pmatrix},
    \label{eq:simpleKraus}
\end{equation}
are expressed in the standard basis $\{\ket{i}\}$ at $\HC_A$ and $\HC_B$, $u
\in [0, 2\pi]$ and $v \in [0, \pi)$. 

While $u$ and $v$ parametrize the channel~\eqref{eq:twoKraus}, they don't
necessarily represent noise parameters that have a monotonic relationship with
the amount of noise introduced by the channel. In certain special cases, noise
parameters can be arrived at intuitively.  For instance when $u = 0$, 
\begin{equation}
    \NC(O) = \cos^2 (\frac{v}{2}) O + \sin^2 (\frac{v}{2}) X O X,
    \label{eq:dephasing}
\end{equation}
is a qubit dephasing channel with dephasing probability $\sin^2
(v/2)$~\footnote{Notice, the dephasing channel is not extremal since each of
its Kraus operators are proportional to a unitary operator~(see discussion
above~\eqref{eq:twoKraus}).}.
By performing a unitary, $X$, at the input channel input $\HC_A$, this
dephasing channel~\eqref{eq:dephasing} can be converted to another dephasing
channel with dephasing probability $1- \sin^2 (v/2)$. Thus a dephasing
probability of half gives maximum dephasing.
This dephasing probability is an intuitive noise parameter in the sense that as
this probability is increased from zero to a half, the channel becomes noisier. 

Another special case is when $u + v = 2 \pi$. Here, if kets $\ket{0}$ and
$\ket{1}$ are interchanged at the channel input and output, $\NC$ becomes a
qubit amplitude damping channel. The qubit amplitude damping channel fixes
$\dya{0}$ but $\dya{1}$ decays to $\dya{0}$ with probability $\sin^2 v$.
Intuitively, this damping probability is a noise parameter in the sense that as
the damping probability is increased from zero to one, the channel becomes
noisier.
Except for these special cases of dephasing and amplitude damping, suitable
noise parameters are not necessarily easy to guess. 

As discussed above, when $\NC$ represents amplitude damping noise, the noise
parameter is the damping probability.  In all other cases, this qubit channel
$\NC$ can be generated from an isometry~(see discussion in
Sec.~\ref{sec:qChan}) of a special form. An isometry of this {\em pcubed}
form~\cite{SiddhuGriffiths16},
\begin{equation}
    V \ket{\al_{i}} = \ket{\bt_{i}} \ot \ket{\gm_{i}},
    \label{eq:pcubeForm}
\end{equation}
where $i \in \{0, 1\}$, takes some special input {\em pure} states
$\{\ket{\al_i}\}$ that are not necessarily orthogonal but form a basis of
$\HC_A$, to {\em product of pure} states $\{ \ket{\bt_i}\}$ at the $\HC_B$
output and $\{ \ket{\gm_i} \}$ at the $\HC_C$ output.
The Gram matrices $G_A, G_B,$ and $G_C$ of $\{ \ket{\al_i}\}, \{ \ket{\bt_j}\}$
and $\{\ket{\gm_{k}} \}$, respectively, satisfy
\begin{equation}
    [G_A]_{ij} = \inpd{\al_i}{\al_j} = \inpd{\bt_i}{\bt_j} \inpd{\gm_i}{\gm_j} 
    = [G_B]_{ij} [G_C]_{ij}
\end{equation}
if and only if $V$ is an isometry, i.e., $V^{\dag} V =
I_A$~\cite{SiddhuGriffiths16}. These matrices take the form
\begin{equation}
    G_A = \begin{pmatrix} 1 & a \\ a & 1 \end{pmatrix}, \quad
    G_B = \begin{pmatrix} 1 & b \\ b & 1 \end{pmatrix}, \quad \text{and} \quad
    G_C = \begin{pmatrix} 1 & c \\ c & 1 \end{pmatrix},
        \label{eq:grmMat}
\end{equation}
where $-1 < a < 1, -1 \leq b \leq 1,-1 \leq c \leq 1$, and $a = bc$.
The parameters $b$ and $c$ completely specify the isometry 
$V$ in~\eqref{eq:pcubeForm} and thus the channel $\NC$. 
One may parametrize $\ket{\al_i}$ using the standard basis as
\begin{equation}
    \ket{\al_i} = \sqrt{\frac{1 + a}{2}} \ket{0} + (-1)^i \sqrt{\frac{1-a}{2}}
    \ket{1}.
    \label{eq:alKets}
\end{equation}
In this parametrization replacing $a$ with $b$ gives $\ket{\bt_i}$ and
replacing $a$ with $c$ gives $\ket{\gm_i}$.
The parameters $b$ and $c$ are related to $u$ and $v$ in~\eqref{eq:simpleKraus}
as follows,
\begin{equation}
    \sin^2 v = \frac{1 - c^2}{1 - (bc)^2}, 
    \quad \text{and} \quad
    \cos^2 u = \frac{1 - b^2}{1 - (bc)^2},
\end{equation}
where $|bc| \neq 1$. The Kraus operators in~\eqref{eq:simpleKraus} can be
written as
\begin{equation}
    K_0 = \begin{pmatrix} 
            \sqrt{\frac{(1+b)(1+c)}{2(1+bc)}} & 0 \\
            0 & \sqrt{\frac{(1-b)(1+c)}{2(1-bc)}}
          \end{pmatrix} 
  \quad \text{and} \quad
    K_1 = \begin{pmatrix} 
            0 & \sqrt{\frac{(1+b)(1-c)}{2(1-bc)}} \\
            \sqrt{\frac{(1-b)(1-c)}{2(1+bc)}} & 0
          \end{pmatrix}.
    \label{eq:krsNewPara}
\end{equation}
While these Kraus operators look more complicated than those
in~\eqref{eq:simpleKraus}, several other channel properties simplify when using
the parameters $b$ and $c$. 
For instance, the channel $\NC$ with parameters $b$ and $c$ is degradable if
$|b/c| < 1$, otherwise $|b/c| \geq 1$ and the channel is
anti-degradable~\cite{SiddhuGriffiths16}.

In general, $-1 \leq b \leq 1$ and $-1 \leq c \leq 1$, however one can simplify
the parameter space. In the discussion above, replacing $b$ with $-b$ while
keeping $c$ fixed results in a new channel $\tilde{\NC}$ which is equivalent to
$\NC$ up to local unitaries at the channel input and output.
To see this, notice this replacement defines a new isometry $\tilde{V}$ of the
pcubed form,
\begin{equation}
    \tilde{V} \ket{\tilde{\al}_i} =\ket{\tilde{\bt}_i} \ot \ket{\gm_i},
    \label{eq:KTilde}
\end{equation}
where $\ket{\tilde{\al}_i}$ and $\ket{\tilde{\bt}_i}$ are kets obtained from
$\ket{\al_i}$ and $\ket{\bt_i}$~(see definition below eq.~\eqref{eq:alKets}) by
replacing $a$ and $b$ with $-a$ and $-b$, respectively.
This new isometry $\tilde{V}$ is related to $V$ in~\eqref{eq:pcubeForm},
via local unitaries as follows, 
\begin{equation}
    (I_C \ot X_B) \tilde{V} = V X_A,
    \label{eq:KTildeK}
\end{equation}
where $X$ is defined in~\eqref{eq:pauliOp}.
In a similar vein, a channel with parameters $b$ and $c$ is equivalent up to local
unitaries to a channel with parameters $b$ and $-c$. These equivalences allow
us to restrict the parameter space $-1 \leq b \leq 1$ and $-1 \leq c \leq 1$ to 
the positive quadrant $0 \leq b \leq 1$ and $0 \leq c \leq 1$.

We show that any channel $\NC$ with parameters $b$ and $c$ can simulate another
channel $\NC'$ with parameters $b$ and $c' \leq c$, in the sense,
\begin{equation}
    \NC' = \NC \circ \MC,
    \label{eq:composition}
\end{equation}
where $\MC$ is a quantum channel. 
Proof of the above equation is easy to see from a pcubed point of view. Let
$\NC: \LC(\HC_A) \mapsto \LC(\HC_B)$ be generated by the isometry
in~\eqref{eq:pcubeForm}, $\NC': \LC(\HC_A) \mapsto \LC(\HC_B)$ be generated by
an isometry $V':\HC_A \mapsto \HC_B \ot \HC_{C'}$ of the same form as $V$
in~\eqref{eq:pcubeForm}, however
\begin{equation}
    V'\ket{\al'_i} = \ket{\bt_i} \ot \ket{\gm'_i}
\end{equation}
where $c' = \inpd{\gm'_0}{\gm'_1}$ and $a' = \inpd{\al'_0}{\al'_1} = bc'$.
The $\MC: \LC(\HC_A) \mapsto \LC(\HC_A)$ channel in~\eqref{eq:composition} is
generated by an isometry $W : \HC_A \mapsto \HC_A \ot \HC_D$ of the
form~\eqref{eq:pcubeForm} with
\begin{equation}
    W \ket{\al'_{i}} = \ket{\al_i} \ot \ket{\dl_i},
\end{equation}
where $d := \inpd{\dl_0}{\dl_1} = c'/c$ takes values between $0$ and $1$
since $0 \leq c' \leq c$.
The relationship in~\eqref{eq:composition} ensures that $\NC'$ is noisier than
$\NC$. As a result, for fixed $b$, if one decreases $c$ then the channel $\NC$
becomes noisier.

This parameter $c$ captures lack of distinguishability between pure states
being arriving at the environment. If $c$ is decreased, more information flows
to the environment. The no-cloning
theorem~\cite{Park70,Dieks82,WoottersZurek82} indicates that such a flow to the
environment must come at the cost of information flow to the output. Thus $\NC$
becomes noisier with decreasing $c$.
We shall be interested in using $c$ as the noise parameter with $b$ fixed.
In the limiting $b = 0$ case, $\NC$ becomes the qubit dephasing
channel~\eqref{eq:dephasing} with dephasing probability $(1-c)/2$.  Here,
decreasing $c$ from $1$ to $0$ increases the dephasing probability from $0$ to
half.

\section{Optimal entanglement sharing}

\subsection{High fidelity entanglement}
\label{sec:BasicQ}

Consider two parties Alice and Bob, connected by some quantum channel $\NC:
\LC(\HC_A) \mapsto \LC(\HC_B)$ where $\HC_A$ and $\HC_B$ have the same
dimension $d$.  Suppose Alice has access to a second $d$-dimensional system
with Hilbert space $\HC_{R}$. 
What bipartite state $\rho_{RA}$ should Alice prepare such that
sharing with Bob one half of this state across the channel $\NC$ results in a
state $\rho_{RB}$ with highest fidelity $F(\rho_{RB},\phi_{RB})$ to a maximally
entangled state,
\begin{equation}
    \ket{\phi}_{RB} = \frac{1}{\sqrt{d}} \ket{\gm}_{RB},
    \label{eq:maxEnt}
\end{equation}
between reference $\HC_R$ and output $\HC_B$?  The optimal state prepared by
Alice, which we denote by $\Lm_{RA}$, and the maximum fidelity,
\begin{equation}
    \OC(\NC) := F(\Lm_{RB},\phi_{RB}),
    \label{eq:maxFID}
\end{equation}
have been characterized previously in terms of the channel's
Choi-Jamio\l{}kowski operator~\cite{VerstraeteVerschelde03,
PalBandyopadhyayEA14,PalBandyopadhyay18} when $\rho_{RA}$ is pure. 
For possibly mixed $\rho_{RA}$, our reformulation of these results in terms of
the standard Kraus decomposition of a channel and the operator norm of the
channel's Choi-Jamio\l{}kowski operator agree with these previous
characterizations.  We extend these results by finding families of mixed input
states $\Lm_{RB}$ that achieve $\OC(\NC)$. This reformulation and extension is
used later in our study. We begin our reformulation using a semi-definite
program 
\begin{align}
    \label{opt:maxFPhi}
    \begin{aligned}
        \text{maximize} \; & F(\rho_{RB},\phi_{RB}) \\
        \text{subject to} \; & \rho_{RB} = \IC_R \ot \NC (\rho_{RA}), \\
        \; & \rho_{RA} \succeq 0, \\
        \; & \Tr (\rho_{RA}) = 1.
    \end{aligned}
\end{align}
The optimum value of the above program gives $\OC(\NC)$ and the density
operator which achieves this optimum gives $\Lm_{RA}$. The following Theorem
captures the solution to the above problem.
\begin{theorem}
    Given a channel $\NC$ with standard Kraus operators $\{K_i\}$, 
    \begin{equation}
        \OC(\NC) = \frac{1}{d}\inpV{K_0}{K_0} 
        = \frac{1}{d} || J^{\NC}_{RB}|| 
        = F(\Lm_{RB}, \phi_{RB}),
    \end{equation}
    where the input $\Lm_{RA}$ has support in the span of
    $\{\ket{K_i^{\dag}}_{RA}\}$ satisfying $\inpV{K_i}{K_i} = \inpV{K_0}{K_0}$.
    \label{th:thOne}
\end{theorem}
\begin{proof}
    Using eq.~\eqref{eq:Fidelity} along with the fact that $\phi_{RB}$ is a
    pure state, one writes $F(\rho_{RB}, \phi_{RB})$ as an inner product
    $\inpV{\rho_{RB}}{\phi_{RB}}$.  This inner product is re-written as $\inpV{
        \IC_R \ot \NC (\rho_{RA})}{\phi_{RB}}$ using the first equality
    constraint in~\eqref{opt:maxFPhi}.  This re-writing can be reduced to
    $\inpV{\rho_{RA}}{ (\IC_R \ot \NC)^{\dag} (\phi_{RB})}$ using
    definition~\eqref{eq:dualInp} of the dual channel. Using discussion
    below~\eqref{eq:dualInp}, or otherwise, one can show that the dual of the
    tensor product of two channels is the tensor product of the dual of
    individual channels. Thus $\inpV{\rho_{RA}}{ (\IC_R \ot \NC)^{\dag}
    (\phi_{RB})} = \inpV{\rho_{RA}}{ \IC_R \ot \NC^{\dag} (\phi_{RB})}$, where
    we used that fact that $\IC_R^{\dag}$ is $\IC_R$. Next, notice $(\IC_R \ot
    \NC^{\dag} ) \phi_{RB}$ is just $\JC_{RA}^{\NC^{\dag}}/d$~\eqref{eq:cj}.
    Using these observations, re-write~\eqref{opt:maxFPhi} as
    \begin{align}
        \label{opt:maxFPhi-2}
        \begin{aligned}
            \text{maximize} \; & 
            \frac{1}{d}\inpV{\rho_{RA}}{\JC_{RA}^{\NC^{\dag}}} \\
            \text{subject to} \; & \rho_{RA} \succeq 0, \\
            \; & \Tr (\rho_{RA}) = 1.
        \end{aligned}
    \end{align}
    Solution to this semi-definite program is $(1/d)$ times the maximum
    eigenvalue of $J^{\NC^{\dag}}_{RA}$ obtained by setting $\rho_{RA} =
    \Lm_{RA}$ where $\Lm_{RA}$ is any density operator with support on the
    eigenspace of this maximum eigenvalue. 
    This largest eigenvalue can be written as $\inpV{K_0^{\dag}}{K_0^{\dag}} =
    \inpV{K_0}{K_0}$ using~\eqref{eq:canKraus} and \eqref{eq:choiDcan}. The
    largest eigenvalue can also be written as the spectral norm, $||
    J^{\NC}_{RB} ||$, by applying definition~\eqref{eq:HNorm}.
    The support of the largest eigenvalue, $\inpV{K_0}{K_0}$, of
    $J^{\NC^{\dag}}_{RA}$ is the span of the collection of eigenvectors
    corresponding to this value. This collection contains eigenvectors
    $\ket{K_i^{\dag}}$ of $J^{\NC^{\dag}}_{RA}$~(see~\eqref{eq:choiD}) with
    eigenvalue $\inpV{K_i^{\dag}}{K_i^{\dag}}$ equaling the largest eigenvalue
    $\inpV{K_0^{\dag}}{K_0^{\dag}}$. The eigenvalues of $J^{\NC^{\dag}}_{RA}$
    can be shown to equal corresponding eigenvalues of $J^{\NC}_{RA}$
    using~\eqref{eq:choiDchoi}, i.e., one can show that
    $\inpV{K_i^{\dag}}{K_i^{\dag}} = \inpV{K_i}{K_i}$.
\end{proof}

The fidelity between a fixed state $\rho_{AB}$ and a fully entangled state,
maximized over all possible fully entangled states is called the {\em fully
entangled fraction}~\cite{BennettDiVincenzoEA96, HorodeckiHorodeckiEA99}
\begin{equation}
    F_e(\rho_{AB}) =  \underset{{U_A}}{\text{max}} \; 
    F\big( \rho_{AB}, (U_A \ot I_B) \phi_{AB} (U_A \ot I_B)^{\dag} \big),
    \label{eq:fullEntFrac}
\end{equation}
where $U_A$ is a unitary operator on $\HC_A$. 

\begin{lemma}
    The {\em largest} fully entangled fraction obtained by sending one half of
    a mixed state $\rho_{RA}$ across the channel $\NC$, maximized over all
    $\rho_{RA}$ equals $\OC(\NC)$~\eqref{eq:maxFID}.
\end{lemma}
\begin{proof}
Notice that the largest fully entangled fraction can be found by modifying the
    optimization problem~\eqref{opt:maxFPhi} as follows: replace $\phi_{RB}$
    with $\chi_{RB} = (U_R \ot I_B) \phi_{RB} (U_R \ot I_B)^{\dag}$ and
    optimize over both unitary matrices $U_R$ and density operators
    $\rho_{RA}$. Notice, in this larger optimization problem, one can simplify
    the objective function $F(\rho_{RB}, \chi_{RB}) = F(\rho'_{RB}, \phi_{RB})$
    where $\rho_{RB}' = (U_R \ot I_B)^{\dag} \rho_{RB} (U_R \ot I_B)$.  Since
    $\rho_{RB}' = \IC \ot \NC (\rho_{RA}')$, where $\rho_{RA}' = (U_R \ot
    I_A)^{\dag} \rho_{RA} (U_R \ot I_A)$, one can rephrase this optimization at
    hand purely in terms of a single new variable $\rho_{RA}'$, satisfying
    $\rho_{RA}' \succeq 0$ and $\Tr(\rho_{RA}') = 1$. In this rephrasing
    variable $U_R$ no longer participates. However the new problem in terms of
    $\rho_{RA}'$ is identical to~\eqref{eq:maxFID}.
\end{proof}

The above result generalizes to mixed state what was implicitly found for pure
states in the proof of Lemma~2 in~\cite{PalBandyopadhyay18}.

Let $\Lm_{RA}$ be the state in Th.~\ref{th:thOne}. We are interested in the
minimum amount of entanglement over all states of this type. To capture this
minimum, we use entanglement of formation~\eqref{eq:entEntro}. When 
$\Lm_{RA}$ is a unique pure state we write the {\em input entanglement}
\begin{equation}
    \EC(\NC) = S(\sg_A),
    \label{eq:inpEnt1}
\end{equation}
when $\Lm_{RA}$ can be chosen to be mixed, we write
\begin{equation}
    \EC(\NC) = \underset{ \Lm_{RA} }{ \text{min} } \; E_f(\Lm_{RA}),
    \label{eq:inpEnt2}
\end{equation}
where $\Lm_{RA}$ are states in Th.~\ref{th:thOne}.  When $\Lm_{RA}$ can be
chosen to be a separable state, $\EC(\NC) = 0$.

\subsection{Multiplicativity}
\label{sec:addtvty}

Suppose Alice and Bob are connected by two independent channels, that may be
same or different. What state should Alice prepare such that sending one half
of it across the joint channel results in Alice and Bob sharing a joint state
with maximum fidelity to a fully entangled state?
What is this maximum fidelity? 
Can one hope to use correlations across the two channels connecting Alice and
Bob to get more fidelity than what can be achieved without using any
correlation across the channels?
Variants of these natural questions have been asked about transmission of
information across asymptotically many uses of quantum channels. Those
questions have been hard to answer. Here we mathematically formulate and answer
the questions we posed above.

Let the two channels connecting Alice and Bob be $\NC_1: \LC(\HC_{A1}) \mapsto
\LC(\HC_{B1})$ and $\NC_2: \LC(\HC_{A2}) \mapsto \LC(\HC_{B2})$, here $d_{A1} =
d_{B1}$ and $d_{A2}= d_{B2}$. For each channel input $\HC_{A1}$ and $\HC_{A2}$,
define auxiliary spaces $\HC_{R1}$ and $\HC_{R2}$. Let $\IC_{R1}$ and
$\IC_{R2}$ be identity maps on these auxiliary spaces, $\LC(\HC_{R1})$ and
$\LC(\HC_{R2})$, respectively, $\HC_{A} := \HC_{A1} \ot \HC_{A2}, \HC_{B} :=
\HC_{B1} \ot \HC_{B2}, \HC_R = \HC_{R1} \ot \HC_{R2}, \NC = \NC_1 \ot \NC_2$,
and $\IC_R = \IC_{R1} \ot \IC_{R2}$.
If Alice prepares a state which does not correlate inputs to the two channels
$\IC_{R1} \ot \NC_1$ and $\IC_{R2} \ot \NC_2$ then the maximum fidelity with a
fully entangled state across auxiliary space $\HC_R$ and the channel output
$\HC_B$ can be found as follows:
\begin{align}
    \label{eq:joint}
    \begin{aligned}
        \text{maximize} \; &  F(\rho_{RB}, \phi_{RB}) \\
        \text{subject to} \; & \rho_{RB} = (\IC_R \ot \NC) \rho_{RA}, \\
        \; & \rho_{RA} = \rho_{R1A1} \ot \rho_{R2A2} \\
        \; & \rho_{RA} \succeq 0, \\
        \; & \Tr (\rho_{RA}) = 1.
    \end{aligned}
\end{align}
The optimum of the above problem is simply $O(\NC_1) O(\NC_2)$. It is
obtained at $\Lm_{RA} = \Lm_{R1A1} \ot \Lm_{R2A2}$ where $\Lm_{R1A1}$
and $\Lm_{R2A2}$ are optima to optimizations of the form~\eqref{opt:maxFPhi}
for $\NC_1$ and $\NC_2$, respectively.
On the other hand, if Alice prepares a state that may correlate the inputs to
$\IC_{R1} \ot \NC_1$ and $\IC_{R2} \ot \NC_2$ then the maximum fidelity 
$\OC(\NC_1 \ot \NC_2)$ is found by solving~\eqref{eq:joint} without the product
constraint, $\rho_{RA} = \rho_{R1A1} \ot \rho_{R2A2}$. This fidelity 
maximum $O(\NC_1 \ot \NC_2)$ can be higher
\begin{equation}
    O(\NC_1 \ot \NC_2) \geq O(\NC_1) O(\NC_2),
    \label{eq:super-mult}
\end{equation}
since the optimum $O(\NC_1) O(\NC_2)$ of~\eqref{eq:joint} bounds from below the
optimum of~\eqref{eq:joint} without the product constraint, $\rho_{RA} =
\rho_{R1A1} \ot \rho_{R2A2}$. 

\begin{theorem}
    The maximum fidelity $O(\NC_1 \ot \NC_2)$ is multiplicative, i.e., equality
    holds in~\eqref{eq:super-mult}
    \begin{equation}
        O(\NC_1 \ot \NC_2) = O(\NC_1) O(\NC_2),
    \end{equation}
    \label{th:multi}
\end{theorem}

\begin{proof}
    Let $\NC_{1}$ and $\NC_2$ have standard Kraus decomposition $\{J_q\}$
    and $\{K_r\}$, respectively. Using Th.~\ref{th:thOne}, we write
    \begin{equation}
        \OC(\NC_1) = \frac{1}{d_{A1}} \inpV{J_0}{J_0} \quad \text{and} \quad
        \OC(\NC_2) = \frac{1}{d_{A1}} \inpV{K_0}{K_0}.
        \label{eq:N1N2b}
    \end{equation}
    A standard Kraus decomposition $\{L_p\}$ for $\NC_1 \ot \NC_2$
    can be chosen such that each $L_p$ is of the form $J_q \ot K_r$
    for some $q$ and $r$. When $q=r=0$, then $p$ can be chosen to be 0,
    \begin{equation}
        L_0 = J_0 \ot K_0
        \label{eq:l1}
    \end{equation}
    since
    \begin{equation}
    \inpV{L_0}{L_0} = \inpV{J_0}{J_0}\inpV{K_0}{K_0} \geq
        \inpV{J_q}{J_q}\inpV{K_r}{K_r} = \inpV{L_p}{L_p}
        \label{eq:l1Ineq}
    \end{equation}
    for all $q,r$ and corresponding $p$. Using~\eqref{eq:l1},
    and Th.~\ref{th:thOne} on $\NC_1 \ot \NC_2$ gives
    \begin{equation}
        \OC(\NC_1 \ot \NC_2) = \frac{1}{d_{A1}d_{A2}}\inpV{L_0}{L_0}.
    \end{equation}
    The above equality, together with~\eqref{eq:N1N2b} and~\eqref{eq:l1Ineq}
    proves the result. 

    Alternatively, notice
    \begin{equation}
        J^{\NC}_{RB} = \sum_p \dya{L_p} = 
        \sum_{qr} \dya{J_q} \ot \dya{K_r} = J^{\NC_1}_{R1B1} \ot
        J^{\NC_2}_{R2B2}.
    \end{equation}
    where the first equality follows from~\eqref{eq:choiCan}.
    Using Th.~\ref{th:thOne}, write
    \begin{equation}
        \OC(\NC_1) = \frac{1}{d_{A1}} || J^{\NC_1}_{R1B1}||, \quad
        \OC(\NC_2) = \frac{1}{d_{A2}} || J^{\NC_2}_{R2B2}||, \quad \text{and} \quad
        \OC(\NC_1 \ot \NC_2) = \frac{1}{d_{A}}|| J^{\NC}_{RB}||,
        \label{eq:os}
    \end{equation}
    where $d_A = d_{A1} d_{A2}$.  The operator norm is sub-multiplicative~(see
    Sec.1.1.3 in~\cite{Watrous18}), 
    \begin{equation}
        || A B || \leq || A || \cdot || B ||,
    \end{equation}
    it implies 
    \begin{equation}
        || A \ot B || \leq || A || \cdot || B ||.
    \end{equation}
    Using the above equation along with~\eqref{eq:super-mult} 
    and~\eqref{eq:os} also proves the result.
\end{proof}

\section{Applications}
\subsection{Extremal qubit channels}
\label{sec:glnQbitNoise}

A qubit channel $\NC$ has $d_A = d_B = 2$. If the channel has one Kraus 
operator then the channel is simply conjugation with a unitary matrix
and $\OC(\NC) = 1$. The next simplest qubit channel has two Kraus operators,
given in~\eqref{eq:twoKraus}. 
One special case of this channel is the qubit amplitude damping channel. Kraus
operators for this amplitude channel can be written as,
\begin{equation}
    K_0 = \begin{pmatrix}
        1 & 0 \\
        0 & \sqrt{1-p}
    \end{pmatrix}, \quad \text{and} \quad
    K_1 = \begin{pmatrix}
        0 & \sqrt{p} \\
        0 & 0
    \end{pmatrix},
    \label{eq:adKrs}
\end{equation}
where $0 \leq p \leq 1$ is the probability that the state $\dya{1}$ damps to
$\dya{0}$. 
A simple calculation shows that these Kraus operators constitute a standard
Kraus decomposition of $\NC$.  Using this decomposition in Th.~\ref{th:thOne},
we find
\begin{equation}
    \OC(\NC) = 1 - p/2 \quad \text{and} \quad
    \Lm_{RA} = \frac{\dyad{K_0}{K_0}}{\inpV{K_0}{K_0}},
    \label{eq:rhoOptAD}
\end{equation}
a result that agree with~\cite{BandyopadhyayGhosh12}. In general, the amount
of entanglement generated at the input~(see def.  in~\eqref{eq:inpEnt1}),
\begin{equation}
    E(\NC) = h(\frac{1}{2-p}),
    \label{eq:entAD}
\end{equation}
where $h(x):= -x \log x - (1-x) \log (1-x)$ is the binary entropy function with
$\log$ base $2$. This value is nonzero, unless $p = 1$ where $E(\NC) = 0$ and
$\Lm_{RA}$ in~\eqref{eq:rhoOptAD} is a product state.

When the qubit channel $\NC$ with two Kraus operators is not an amplitude
damping channel, the channel Kraus operators take the
form~\eqref{eq:krsNewPara}. These Kraus operators $\{K_0, K_1\}$ have two
parameters $0 \leq b \leq 1$ and $0 \leq c \leq 1$. If $b$ is fixed and
$c$ is decreased from 1 the channel becomes more noisy~(see discussion
containing~\eqref{eq:composition}). Operators $\{K_0,K_1\}$ 
form a standard Kraus decomposition. Using them in Th.~\ref{th:thOne}, gives
\begin{equation}
    \OC(\NC) = \frac{(1+c)(1-b^2c)}{2(1-b^2c^2)}
    \quad \text{and} \quad
    \Lm_{RA} = \begin{cases}
        \frac{\dyad{K_0^{\dag}}{K_0^{\dag}}}{\inpV{K_0^{\dag}}{K_0^{\dag}}} & \text{if} \quad b \neq  1 \; \text{and} \; c \neq 0 \\
        \sum_{ij} f_{ij} \dyad{K_i^{\dag}}{K_j^{\dag}} & \text{if} \quad b = 1 \; \text{or} \; c = 0
    \end{cases}.
    \label{eq:OGln}
\end{equation}
where complex numbers $f_{ij}$ are free except that they result in a valid
density operator $\Lm_{RA}$. At $b = 1$ or $c = 0$, $\Lm_{RA}$ is supported on
a two-dimensional space spanned by $\{\ket{K_0^{\dag}}_{RA},
\ket{K_1^{\dag}}_{RA}\}$. This two-dimensional space is a subspace of a two
qubit space $\HC_{RA}$. Quite generally, such a subspace has at least one
product state~(see Lemma in~\cite{NiuGriffiths98}), but typically there are
two~\cite{RuskaiSzarekEA02, SiddhuGriffiths16}. In the $c=0$ case, these
product states take the simple form
\begin{equation}
    \ket{+}_R \ot \ket{\psi_{+}}_A  \quad \text{and} \quad
    \ket{-}_R \ot \ket{\psi_{-}}_A,
    \label{eq:prods}
\end{equation}
where $\ket{\psi_+}_A = \frac{1}{\sqrt{2}} ( \sqrt{1+b} \ket{0} + \sqrt{1-b} \ket{1})$,
$\ket{\psi_-}_A = \frac{1}{\sqrt{2}} (\sqrt{1+b} \ket{0} - \sqrt{1-b} \ket{1})$,
$\ket{+}_A = \frac{1}{\sqrt{2}} (\ket{0} +\ket{1})$, and
$\ket{-}_A = \frac{1}{\sqrt{2}} (\ket{0} - \ket{1})$.

At $b = 1$ or $c=0$ one can choose $\Lm_{RA}$ to be a projector onto a product
state. As a result, at $b=1$ or $c=0$, the input entanglement, defined
in~\eqref{eq:entEntro0}, is zero. In general,
\begin{equation}
    E(\NC) = \begin{cases}
        0 & \text{if} \quad  b = 1 \; \text{or} \; c=0\\
        h(\frac{(1+b)(1-bc)}{2(1-b^2c)}) & \text{otherwise}\\  
    \end{cases}
    \label{eq:entGln}
\end{equation}
where expressions for $\EC(\NC)$ at $b \neq 1$ and $c \neq 0$ comes from using
the form of $\Lm_{RA}$ in~\eqref{eq:OGln}. 
In Fig.~\ref{fig:ent-qbit} we fix $b$ and plot $E(\NC)$ as a function of $c$;
increasing $c$ makes $\NC$ less noisy~(see discussion containing
eq.~\eqref{eq:composition}).
In these plots, as $c$ is increased from zero, the value of $\EC(\NC)$
dis-continuously increases from $0$, at $c = 0$, and continues to monotonically
increase until $c=1$, where $\NC$ becomes a perfect channel.
Across various plots with fixed $b$, we notice increasing $b$ decrease
$\EC(\NC)$, which ultimately goes to zero as $b \mapsto 1$ for all $bc \neq 1$.

All these features mentioned above are intriguing. In the parameter range $0 <
c < 1$, one finds an expected result~\cite{PalBandyopadhyayEA14} that the
minimum amount of entanglement at the input to have maximum fidelity with a
fully entangled output is strictly less than one. In particular, if one
generates more than $\EC(\NC)<1$ entanglement at the input, the fidelity with a
maximally entangled output is strictly less. The key addition here is the
quantification of the amount of entanglement and a parametrization of the
channel in such a way that the amount of entanglement is monotone in the noise
parameters of the channel.

Next, at $c = 0$, there is a discontinuous change in $\EC(\NC)$ which starts at
zero and then takes a large finite value $\simeq h\big( (1+b)/2 \big)$. From a
mathematical standpoint, the discontinuity arises because the solution to the
optimization~\eqref{opt:maxFPhi} becomes degenerate and this degeneracy allows
more freedom in choosing optimum inputs. Due to the structure of qubit
channels, this input can be chosen to be separable, as mentioned in the
discussion containing~\eqref{eq:prods}.
\begin{figure}[ht]
	\centering
    \includegraphics[]{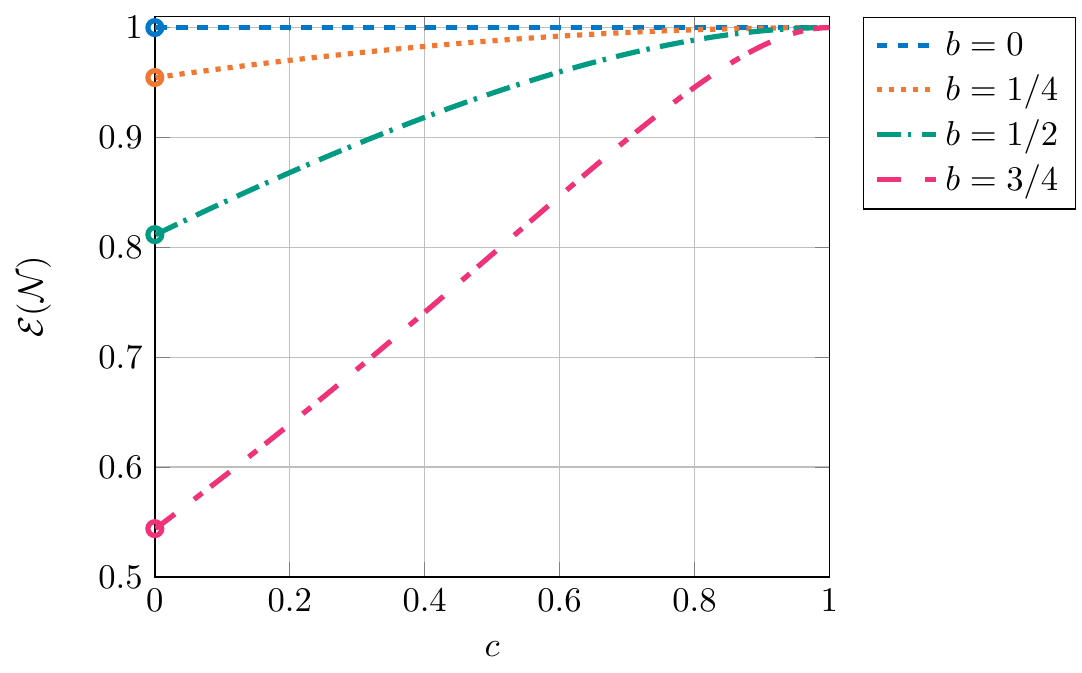}
    \caption{Plot of $\EC(\MC)$ as a function of $c$ for various $b$ values.
    The open circle indicates that the value is zero.}
	\label{fig:ent-qbit}
\end{figure}

\begin{figure}[ht]
	\centering
    \includegraphics[]{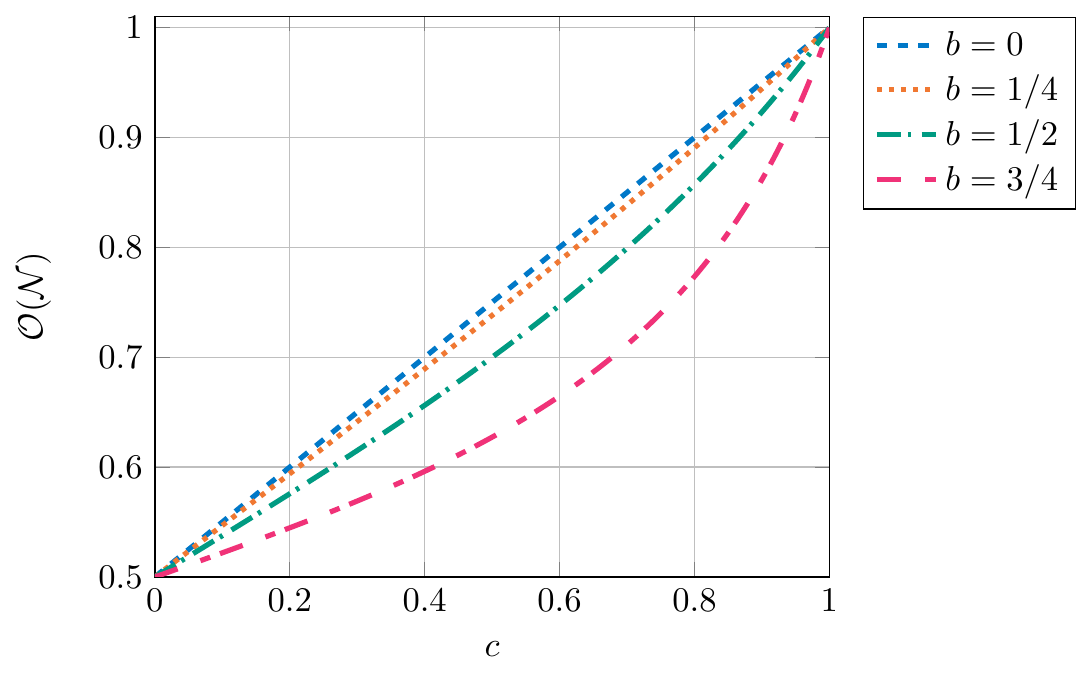}
    \caption{Plot of $\OC(\MC)$ as a function of $c$ for various $b$ values.}
	\label{fig:o-qbit}
\end{figure}

\subsection{Qubit Pauli channels}
\label{sec:pauliCase}

A qubit Pauli channel $\NC: \HC_A \mapsto \HC_B$ can be written as
\begin{equation}
    \NC(\rho) = \sum_i p_i \sg_i \rho \sg_i^{\dag},
    \label{eq:PauliChannel}
\end{equation}
where $p_i \geq 0$, $\sum_i p_i = 1$, and the Kraus operators 
$\{ \sqrt{p_i}\sg_i \}$,  $\sg_i: \HC_A \mapsto \HC_B$, 
are proportional to Pauli matrices. These matrices can be written in the
standard $\{ \ket{0}, \ket{1} \}$ basis of $\HC_A$ and $\HC_B$ as
\begin{equation}
    \sg_0 = I = \begin{pmatrix}
        1 & 0 \\
        0 & 1\\
    \end{pmatrix},
    \quad
    \sg_1 = X = \begin{pmatrix}
        0 & 1 \\
        1 & 0\\
    \end{pmatrix},
    \quad
    \sg_2 = Y = \begin{pmatrix}
        0 & -i \\
        i & 0\\
    \end{pmatrix},
    \quad \text{and} \quad
    \sg_3 = Z = \begin{pmatrix}
        1 & 0 \\
        0 & -1\\
    \end{pmatrix}.
\end{equation}
Without loss of generality we can assume $p_0 \geq p_i$ for all $i \in
\{1,2,3\}$.  This assumption comes from the following argument. Assume $p_i
\geq p_j$ for some $i \neq 0$ and all $j \in \{0,1,2,3\}$, then conjugating the
input $\rho$ with $\sg_i$ will still result in a Pauli
channel~\eqref{eq:PauliChannel}. However, this resulting channel will have $p_0
\geq p_i$ for all $i \in \{1,2,3\}$.

For qubit Pauli channels, the value of $\OC(\NC)$ and the fact that it can be
achieved using a maximally entangled input state $\Lm_{RA}$ was found
in~\cite{PalBandyopadhyayEA14}, however we note later that one can also achieve
$\OC(\NC)$ using a separable pure state when $\NC$ is very noisy.
Since the Pauli matrices are orthogonal to each other, in a standard Kraus
decomposition, $\{K_i\}$, of $\NC$ we can always chose each $K_i$ to be
$\sqrt{p_j} \sg_j$ for some $j$. As $p_0 \geq p_i$, $K_0 = \sqrt{p_0} \sg_0$
and from Th.~\ref{th:thOne} we get
\begin{equation}
    \OC(\NC) = p_0.
    \label{eq:Pauli}
\end{equation}
When $p_0 > p_j$ for all $j$, $\Lm_{RA} = \dya{\sg_0^{\dag}}/2$, i.e.,
$\Lm_{RA}$ is a projector onto a maximally entangled state and thus $\EC(\NC) =
1$. However, if for some $i$, $p_0 = p_i$ then $\Lm_{RA}$ is any density
operator with support in a space spanned by $\{ \ket{\sg_0^{\dag}}_{RA},
\ket{\sg_i^{\dag}}_{RA} \}$. This space is a two-dimensional subspace of a two
qubit space $\HC_{RA}$. Following the discussion containing~\eqref{eq:prods},
this subspace contains at least one product state. As a result, for any $i$ if
$p_0 = p_i$ we can choose $\Lm_{RA}$ to be a product state and thus $\EC(\NC) =
0$.  Consequently,
\begin{equation}
    \EC(\NC) = 
    \begin{cases}
        1 & \text{if} \quad p_0 > p_i \; \forall i, \\
        0 & \text{if} \quad p_0 = p_i \; \text{for some} \; i.
    \end{cases}.
    \label{eq:entPauli}
\end{equation}
When $p_0 = p_i$, the best fidelity with a maximally entangled state at the
output is achieved by sending a separable input $\Lm_{RA}$.  Consequently, the
output $\Lm_{RB}$ is also separable. This separable output is expected to have
a small fidelity with a fully entangled state. This expectation is met, the
condition $p_0 = p_i$ together with $\sum_i p_i = 1$ forces $p_0 \leq 1/2$, and
thus $\OC(\NC) \leq 1/2$. Such a value of half for fidelity with a maximally
entangled state $\ket{\phi}_{AB}$ is considered small since this value of half
can be achieved by a simple separable state $\rho_{RB} = \frac{1}{2} (\dya{00}
+ \dya{11})$.

One may wonder which qubit Pauli channels satisfy $p_0 = p_i \geq p_j$. 
Any qubit Pauli channel of this type is anti-degradable.  In general, $\NC$
in~\eqref{eq:PauliChannel} with $p_0 \geq p_i$ is
anti-degradable~\cite{NiuGriffiths98, Cerf00, CubittRuskaiEA08} if and only if
\begin{equation}
    p_1 + p_2 + p_3 + \sqrt{p_1p_2} + \sqrt{p_1p_3} + \sqrt{p_2p_3} \geq 1/2.
    \label{eq:anti-Deg}
\end{equation}
We are interested in the case where $p_0 = p_i$ for some $i$. The above
condition remains unaffected when permuting $p_i$ and $p_j$, thus we let $p_0 =
p_1 = p$, denote $p_2$ by $q$ then $p_3 = 1-2p-q$. Using these substitutions on
the left side of~\eqref{eq:anti-Deg}, together with $1 \geq p \geq q \geq 0$
and $p \geq 1 - 2p - q$ we find that the above inequality~\eqref{eq:anti-Deg}
is always satisfied.  Thus $p_0 = p_i \geq p_j$ implies that the qubit Pauli
channel $\NC$ is anti-degradable.

Pauli channels~\eqref{eq:PauliChannel} have a key property, up to local
unitaries at the channel input and output, any unital qubit channel can always
be written as a Pauli channel~\cite{RuskaiSzarekEA02}. An interesting
observation about qubit channels is that $\Lm_{RA}$ in Th.~\ref{th:thOne} can
be chosen to be a maximally entangled state if and only if $\NC$ is
unital~\cite{PalBandyopadhyayEA14}. It is interesting for that reason to ask if
such a result holds in higher dimension. In this next section, we find that it
doesn't. 

In the case of qubit Pauli channels, but also for extremal qubit channels, we
found that it is possible to find separable input states $\Lm_{RA}$ that
achieve the most fidelity with a fully entangled state at the channel output.
This separable state appeared when a qubit channel $\NC$'s standard Kraus
decomposition $\{K_i\}$ satisfied the condition $\inpV{K_0}{K_0} =
\inpV{K_j}{K_j}$, for at least one $j \neq 0$. Using eq.~\eqref{eq:choiCan},
this condition reduces to the channel's Choi-Jamio\l{}kowsi operator
$J^{\NC}_{RB}$ having its largest eigenvalue be degenerate. In general, we have
the following lemma.

\begin{lemma}
    \label{lm:qbit}
    If $\NC$ is a qubit channel and the largest eigenvalue of $J^{\NC}_{RB}$
    is degenerate, then $\Lm_{RA}$ in Th.~\ref{th:thOne} can be chosen to be
    separable. 
\end{lemma}
\begin{proof}
    Let $\{K_i\}$ be a standard Kraus decomposition of $\NC$.  Since
    $J^{\NC}_{RB}$ is degenerate, $\inpV{K_0}{K_0} = \inpV{K_1}{K_1}$ and
    $\Lm_{RA}$ has support in the span of $\{ \ket{K_0^{\dag}},
    \ket{K_1^{\dag}} \}$.  This support is a two-dimensional subspace of a
    two-qubit space, and thus contains a product state. Hence $\Lm_{RA}$ can be
    chosen to be a projector onto this product state.
\end{proof}

While it may be tempting to conjecture that the above result holds in higher
dimensional channels, we show in the next section that it doesn't.

\subsection{Some qutrit channels}
\label{sec:qtrit}

We construct two qutrit channels. The first channel, $\MC$, is not unital but
its optimal input state $\Lm_{RA}$, defined in Th.~\ref{th:thOne}, is unique
and maximally entangled.  The second channel, $\PC$, is unital, however its
optimal input state $\Lm_{RA}$ is neither maximally entangled nor separable.
Using the second channel, we demonstrate that when the largest eigenvalue of
$J^{\NC}_{RB}$ is degenerate, $\Lm_{RA}$ can still be entangled. The
demonstration contrasts with Lemma~\ref{lm:qbit}.

Let $\HC_A$ and $\HC_B$ be three-dimensional Hilbert spaces. Let $\MC:
\LC(\HC_A) \mapsto \LC(\HC_B)$ be a channel with Kraus operators
\begin{equation}
    K_0 = \sqrt{\lm} I, \quad
    K_1 = \sqrt{1-\lm} (\dyad{0}{1} + \dyad{1}{0}), \quad \text{and} \quad
    K_2 = \sqrt{1-\lm} \dyad{1}{2},
    \label{eq:krsQtrit}
\end{equation}
where $0 \leq \lm \leq 1$. This channel $\MC$ is not unital, except when $\lm =
1$. When $2/5 < \lm < 1$, $\{K_i\}$ is a standard Kraus decomposition of $\MC$
with $\inpV{K_0}{K_0} > \inpV{K_i}{K_i}$ for all $i \neq 0$. From
Th.~\ref{th:thOne} we find
\begin{equation}
    \OC(\MC) = \lm, \quad \Lm_{RA} = \frac{1}{3}\dya{I}, \quad
    \text{and} \quad \EC(\MC) = \log_2 3.
\end{equation}
Thus when $2/5 < \lm < 1$, the input $\Lm_{RA}$ is unique, and it is maximally
entangled, however the channel $\MC$ is not unital.

Let $\PC: \LC(\HC_A) \mapsto \LC(\HC_B)$ be a qutrit channel with Kraus
operators
\begin{align}
    \begin{aligned}
        L_0 &= \sqrt{\frac{z+2}{4}} \big( \dyad{0}{1} + \dyad{1}{0} \big),
    \quad
        L_1 = \sqrt{\frac{1-z}{2}} \big( \dyad{1}{2} + \dyad{2}{1} \big),
    \\
        L_2 &= \sqrt{\frac{1-z}{2}} \big( \dyad{0}{2} + \dyad{2}{0} \big),
    \quad \text{and} \quad
        L_3 = \sqrt{\frac{z}{4}} \big( \dya{0} + \dya{1} - 2\dya{2} \big),
    \end{aligned}
    \label{eq:krsQtrit2}
\end{align}
where $0 \leq z \leq 1$. Since each Kraus operator $L_i$ is Hermitian, $\PC$ is
unital~(see discussion below~\eqref{eq:kraus}). Kraus operators $\{L_i\}$ are
standard and thus Th.~\ref{th:thOne} immediately gives $\OC(\MC) = (z+2)/6$.
When $z \neq 0$,
\begin{equation}
    \Lm_{AR} = \dya{L_0^{\dag}}
\end{equation}
where $\ket{L_0^{\dag}}_{RA} = \frac{1}{\sqrt{2}} (\ket{01} + \ket{10})$ is not
a maximally entangled state of two qutrits. When $z = 0$, $L_3 = 0$, $\inp{L_0}
= \inp{L_1} = \inp{L_2}$ and thus largest eigenvalue of $J^{\MC}_{RB}$ has a
degenerate spectrum. In this case, $\Lm_{RA}$ has support in a subspace $\SC$
spanned by $\{ \ket{L_0^{\dag}}_{RA}, \ket{L_1^{\dag}}_{RA},
\ket{L_2^{\dag}}_{RA} \}$. 
This subspace only contains non-product vectors, i.e., it is {\em completely
entangled} in the sense of Parthasarathy~(see Def. 1.2
in~\cite{Parthasarathy04}).  Consequently, any density operator $\Lm_{RA}$
supported on this subspace is entangled.

\section{Discussion}

In this work we considered a one-shot setting where one half of any bipartite
mixed state may be sent across a single use of a fixed channel $\NC$. 
The goal in this setting is to share a state with maximum fidelity $\OC(\NC)$
to a fully entangled state.
Interestingly, maximum fidelity $\OC$ defined in the one-shot setting fully
characterizes the ability of any channel to share high fidelity entanglement
over multiple channel uses, possibly used in parallel with other channels.
This extension follows from multiplicative nature of $\OC$, proved in
Sec.~\ref{sec:addtvty}.

Using a semi-definite program, we reformulate the maximum fidelity, found
previously for pure state inputs~\cite{PalBandyopadhyay18,
PalBandyopadhyayEA14, VerstraeteVerschelde03}.  
The first reformulation, see Theorem~\ref{th:thOne} and its proof, makes
greater use of a channel's Kraus operators rather than its Choi-Jamio\l{}kowski
operator, as done previously.
In particular, optimal input(s) achieving $\OC(\NC)$ are simply linear
combinations of flattened versions of a channel's standard Kraus operators with
largest norm, and the optimal value $\OC(\NC)$ is this largest norm itself.
These two channel representations are formally equivalent~(see
Sec.~\ref{sec:qChan} for brief discussion), however the Kraus decomposition can
sometimes be easier to work with and can provide different insights when
discussing maximum fidelity $\OC(\NC)$, but perhaps in other cases as well.
In the present case, the standard Kraus operators~(see Sec.~\ref{sec:standard}
for definition) simplifies the search for and broadens the types of channel
inputs $\Lm_{RA}$ which achieve $\OC$.

One way in which we have broadened the search for optimal inputs $\Lm_{RA}$ is
to identify channels $\NC$ for which $\Lm_{RA}$ can be chosen to be separable.
This choice appears in two notable cases. First, when $\NC$ is an extremal
qubit channel. Here, separability of $\Lm_{RA}$ leads to a discontinuous jump
in the minimal amount of entanglement $\EC(\NC)$ generated to achieve maximum
fidelity with a fully entangled state~(see discussion with
Fig~\ref{fig:ent-qbit}).
A second notable case where $\Lm_{RA}$ can be chosen to be separable is for
noisy unital qubit channels where the input may be ordinarily chosen to be
fully entangled~(see discussion containing eq.~\eqref{eq:entPauli}).
These findings motivate a characterization of channels $\NC$ for which
$\Lm_{RA}$ is possibly separable, i.e., $\EC(\NC) = 0$.  One typically expects
such channels to not be useful for sharing entanglement in the type of one-shot
setting discussed in Sec.~\ref{sec:BasicQ}. One example of such channels is in
Lemma~\ref{lm:qbit}. The lemma extends to channels with Choi-Jamio\l{}kowsi
operator $J^{\NC}_{AB}$ having a greater than $(d-1)^2$ fold degeneracy in
their largest eigenvalue. The support of this largest eigenvalue subspace
always has a product state~(proof for this can be constructed using Prop~1.4
in~\cite{Parthasarathy04}) and thus, $\Lm_{RA}$ can be chosen to be a product
state and $\EC(\NC) = 0$. On the other hand, we also find a channel whose
Choi-Jamio\l{}kowsi operator has a degeneracy in its largest eigenvalue but the
optimal input for the channel must be entangled.
 
Another way in which we have broadened the search for optimal inputs $\Lm_{RA}$
is to consider extension of results found previously.  For qubit channels, a
fully entangled input was known to achieve $\OC$ if and only if the channel was
unital. In higher dimensions, we find this result no longer holds. We construct
a unital qutrit channel for which the optimal input must be less than fully
entangled. We also construct a qutrit channel which is not unital, but for
which a fully entangled input is necessary to obtain the largest overlap.

Our second reformulation of $\OC(\NC)$ in Theorem~\ref{th:thOne} notes that it
equals the operator norm of the channel's Choi-Jamio\l{}kowski operator, upto
normalization.  This observation can not only simplify discussions about
$\OC(\NC)$~(for instance see proof of Th.~\ref{th:multi}), it also gives the
operator norm of the Choi-Jamio\l{}kowski operator a simple interpretation.

The single channel use setting discussed here can be extended by allowing the
reference system and the channel output system to be processed using local
operations and one-way or two-way classical communication, labeled 1-LOCC and
2-LOCC respectively.
Building on ideas in~\cite{VerstraeteVerschelde02,VerstraeteVerschelde03a}, it
has been shown for qubit channels that maximum fully entangled fraction
allowing a single round of 2-LOCC, $\OC'$, equals
$\OC$~\cite{PalBandyopadhyayEA14}. Understanding $\OC'$ in higher dimensional
channels while exploring optimal protocols and multiplicativity of $\OC'$ may
form an interesting direction of future work.  
Extending our work to a setting where the reference system also becomes noisy
may be interesting.  Prior discussions~\cite{FilippovRybarEA12,
FilippovFrizenEA18} on this setting connect with entanglement annihillating
channels~\cite{MoravcikovaZiman10}.
Another direction can come from
extending results in Sec.~\ref{sec:pauliCase} where we show that that a set of
qubit Pauli channels with $\EC(\NC) = 0$ also have no quantum capacity $\QC$.
It could be interesting to study the relation of $\OC$ and $\EC$ to $\QC$.

Along the way to analyzing the maximum fidelity, we found it useful to study
extremal qubit channels. These simple channels can be considered the most basic
qubit channels. However, to our knowledge, noise parameters for these channels
have not been adequately discussed. In Sec.~\ref{sec:extremeQubit}, we show the
pcubed point of view allows one to identify noise parameters for this channel
in such a way that channel becomes demonstrably noisier as a parameter is
varied monotonically. Hope is that such identification makes this channel class
a better understood and non-trivial test-bed for ideas in quantum information
science. We also flesh out two useful properties of general channels. First, in
Sec.~\ref{sec:standard} the existence of a standard Kraus decomposition where
the Kraus operators are orthogonal and their norm is ordered. Second, in
Sec.~\ref{sec:dual}, we show how the Choi–Jamio\l{}kowski operator of a channel
and its dual can always be taken to be complex conjugates of each other.

\section{Acknowledgements}
VS thanks Felix Leditzky for helpful discussions, Sergey Filippov for bringing
ref.~\cite{FilippovRybarEA12, FilippovFrizenEA18} to his attention, Chloe Kim
and Dina Abdelhadi for useful comments.

\end{document}